\documentclass[11pt]{article}

\usepackage{graphics}
\usepackage{amssymb}
\usepackage{amsmath}
\usepackage{amsthm}

\usepackage[colorlinks=false]{hyperref}
\usepackage[pdf]{pstricks}
\usepackage{caption}
\usepackage{subcaption}
\usepackage{graphicx}

\newtheorem{Definition}{Definition}
\newtheorem{Lemma}{Lemma}
\newtheorem{Proposition}{Proposition}

\title{Stability of Breathers for a Periodic Klein-Gordon Equation}

\author{Martina Chirilus-Bruckner\footnote{Mathematisch Instituut - Universiteit Leiden, P.O. Box 9512, 2300 RA, Leiden, The Netherlands; m.chirilus-bruckner@math.leidenuniv.nl}, Jes\'us Cuevas-Maraver\footnote{Grupo de F\'{i}sica No Lineal, Departamento de F\'{i}sica Aplicada I, Universidad de Sevilla. Escuela Polit\'{e}cnica Superior, C/ Virgen de \'{A}frica, 7, 41011-Sevilla, Spain; jcuevas@us.es}\footnote{Instituto de Matem\'{a}ticas de la Universidad de Sevilla (IMUS). Edificio Celestino Mutis. Avda. Reina Mercedes s/n, 41012-Sevilla, Spain}, Panayotis G. Kevrekidis\footnote{Department of Mathematics and Statistics, University of Massachusetts Amherst, Amherst, MA 01003-4515, USA; kevrekid@umass.edu}}

\date{}

\begin{document}

\maketitle

\begin{abstract}
\small The existence of breather type solutions, i.e., periodic in time, exponentially localized in space solutions, is a very unusual feature for continuum, nonlinear wave type equations. Following an earlier work [Comm. Math. Phys. {\bf 302}, 815-841 (2011)], establishing a
theorem for the existence of such structures, we bring to bear a combination of analysis-inspired numerical tools that permit the construction of such waveforms to a desired numerical accuracy. In addition, this enables us
to explore their numerical stability. Our computations show that for the spatially heterogeneous form of the $\phi^4$ model considered herein, the breather solutions are generically found to be unstable. Their instability seems to generically favor the motion of the relevant structures. We expect that these results may inspire further studies towards the identification of stable continuous breathers in spatially-heterogeneous, continuum nonlinear wave
equation models.    
\end{abstract}

\section{Introduction}

The sine-Gordon model is a quintessential example of a dispersive partial differential equation model within Nonlinear Science that has been explored in numerous reviews~\cite{KM89}, as well as books~\cite{eilbeck,dauxois,sgbook}. One of the very well-known and exciting features of this integrable (via the inverse scattering~\cite{akns}) equation is the existence of exact breather solutions. These are temporally periodic, exponentially spatially localized waveforms that are known in an explicit analytical form in this model.

The presence of such breathers has been recognized in {\it spatially discrete} models as a rather generic feature, ever since the work
of Sievers-Takeno~\cite{ST88}, Page~\cite{P90} and many others. Indeed, not only has this work spearheaded applications
in areas ranging from optical waveguide arrays, to superconducting Josephson junctions, atomic
condensates, DNA and beyond, but it has also inspired numerous reviews summarizing the pertinent
progress, see, e.g.,~\cite{A97,FG08}.

On the other hand, in a classic paper from 30 years ago, Birnir, McKean and Weinstein showed a quite remarkable result~\cite{BMW94}, namely that
only perturbations of the integrable sine-Gordon (sG) model of the form $\sin(u)$, $u \cos(u)$ and $1+3\cos(u)-  4 \cos(u/2)+ 4
\cos(u) \log(\cos(u/4))$ can give rise to breathing waveforms. The first two of these stem from rescalings of the standard
sG breather, while the 3rd one is believed to be impossible. This suggests that breathers are rather {\it non-generic} in
continuous problems. Indeed, in a sense this is intuitively understandable. 
{On the one hand,} a breather will have an
intrinsic frequency associated with it{; o}n the other hand, the background state on which it lies will have a
continuous spectrum of plane wave excitations. Generically, the intrinsic breather frequency or, most typically,
its (nonlinearity induced) harmonics would find themselves in resonance with the continuous spectrum, opening a channel of ``radiative decay'' for the breather. That is the principal reason why in a ``classic'' model such as the $\phi^4$, one of
the celebrated results concerns the non-existence of breathers in the model (at least in a truly
localized form) on account of such a resonance. That being said, in the sG case, the ``magic'' of integrability
precludes the activation of such resonances and leads into the persistence of the exact breather waveform.

More recently, these classic findings have prompted a renewed interest in seeking to identify continuum (but now heterogeneous)  models in which one can establish  {\it rigorously} the existence of such breather waveforms. This was initiated in the study  of one of the present authors~\cite{BCLS11} and  was continued by other groups via different types of (variational) methods~\cite{HR19}, yet the fundamental principle is clear, namely to construct a heterogeneous problem such that its band structure can be identified and the frequency of the breather
and its potential harmonics is non-resonant with the continuous spectral bands.

It is this vein of research that we bring to bear herein, by complementing it with detailed numerical studies. Indeed, upon selecting an example
that is promising towards the avoidance of relevant resonances, we find the relevant breather waveform numerically. We then verify that indeed per the theoretical prediction, the multiples of the relevant frequency do not collide with the spectral bands. We subsequently perform the spectral stability analysis of the relevant breather waveform and also delve into continuations over different ones among the breather (e.g., frequency) and model (e.g., the periodic potential) parameters. Interestingly, we find that in the setting considered herein, the breather waveform
is spectrally unstable. Nevertheless, when exploring the dynamical evolution of the respective structures, we find that typically the result of the instability is {\it not} the disintegration of the breather, but rather its mobility. This is also, to some degree, surprising, given that
in media where the spatial periodicity is reflected in their ``discreteness'' it is well-known that so-called Peierls-Nabarro barriers hinder breather mobility~\cite{A97,FG08}.

Our presentation of these results is structured as follows: In section 2, we lay out the general setup of the model and also summarize the main theoretical findings. In section 3, we then present our numerical computations for the breathers, their spectral stability, their parametric continuations and their nonlinear dynamics. Finally, in section 4, we summarize our findings and present our conclusions, as well as some
directions for future studies.

\section{Mathematical Setup and Main Theoretical Result}
\label{SEC:EXISTENCE}

In line with some of these above works and the classical results in sG and $\phi^4$ mentioned in the Introduction, in the present work we consider a Klein-Gordon type equation given by
\begin{align}\label{EQ:NLW}
s(x) \partial_t^2 u = \partial_x^2 u - q(x) u + \rho u^3
\end{align}
for $ x, t, u = u(x,t) \in \mathbb{R} $ with spatially $1$-periodic coefficients $s, q$ and some constant $ \rho \in \mathbb{R} $. The focus of the present
study is on the spectral stability of breathers (see Figure~\ref{FIG:breather}).


\begin{Definition}[\underline{Breather solutions}]
\label{DEF:breathers}
We call a solution $u$ of \eqref{EQ:NLW} a breather (solution), if there exist $ \omega_*, \beta > 0 $ such that the following holds.
\begin{itemize}
 \item[(i)] For all $x, t \in \mathbb{R}$, we have $ u(x, t) = u(x, t+ \frac{2 \pi}{\omega_*}) $.  (``Periodicity in time")
 \item[(ii)] For all $t \in \mathbb{R}$, we have $ \lim_{|x| \rightarrow \infty} u(x,t)e^{-\beta |x|}  = 0 $. ((exponential) ``Localization in space")
\end{itemize}
\end{Definition}

\begin{figure}[!t]
\centering
 \includegraphics[width=0.8\textwidth]{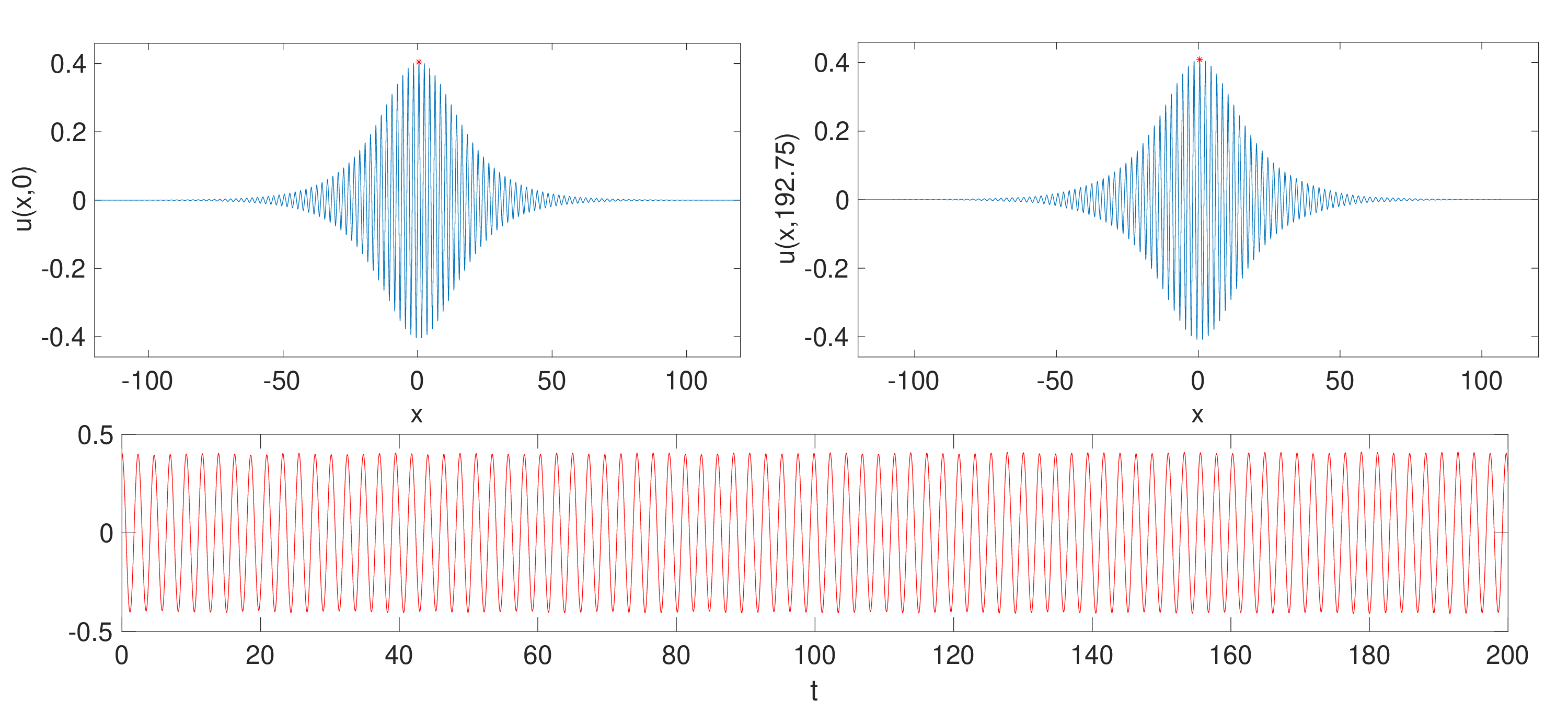}
 \caption{Breather spatial profiles in the upper left and right panels are at different times $t = 0, t \approx 25 \cdot 2\pi/\omega_*$. The lower panel shows the time-evolution of the center value (red dot). Coefficients are chosen to be a {\it resonance-free triplet} $ s_{step}, q_{step}, \omega_* $ as in Definition~\ref{def:resonance_free_triplet} with $ p = 0.43, \varepsilon =  0.15$.
 The initial condition is based on Proposition~\ref{prop:breathers} (Numerical simulation with \texttt{pdepe} from MATLAB). }
\label{FIG:breather}
\end{figure}
The existence of breathers for \eqref{EQ:NLW} has been demonstrated in \cite{BCLS11} for a specific choice of step-function coefficients $ s, q $ (as described in Definition~\ref{def:resonance_free_triplet} in the special case $ p = \frac{6}{13} $). This result seemed rather exceptional given the ``rigidity of breathers'' discussed in the Introduction;
in addition to~\cite{BMW94}, see also \cite{Den93,Den95}.
In fact, the construction of breathers for the case of periodic coefficients uses the tailoring of a spectral picture that would be impossible in the constant-coefficient case.\\
We now build up towards an extended existence result for breathers in the $\phi^4$ model \eqref{EQ:NLW} which is based on the approach in \cite{BCLS11}. The extension manifests itself in determining a whole family of step functions (see Definition~\ref{def:resonance_free_triplet}) for which the core mechanism of the breather construction in \cite{BCLS11} - namely, the tailoring of the band structure - is possible. The corresponding existence result for breathers is then stated in Proposition~\ref{prop:breathers}. Instead of demonstrating its proof we instead give a detailed exposition of the design of the band structure.
\subsection{Designing the band structure to support breather solutions}
Let us now turn our attention to the linear Klein-Gordon equation with periodic coefficients, so
\begin{align}\label{eq:linear_problem}
s(x)\partial_t^2 u = \partial_x^2 u + q(x) u\, ,
\end{align}
for $u = u(x,t)$. The related eigenvalue problem
\begin{align}\label{eq:spectral_problem}
 -\omega^2 s(x) v = \partial_x^2 v + q(x) v\, ,
\end{align}
can be obtained via the ansatz
$$u(x,t) = e^{i\omega t} v(x)\, , \quad  \omega \in \mathbb{R} \, .$$
A value $i\omega \in i\mathbb{R}$ belongs to the spectrum if and only if there exists bounded solutions $v$ of \eqref{eq:spectral_problem}, in other words, when the corresponding Floquet exponents are purely imaginary (see also \eqref{eq:Floquet_stuff} below). It is known from Floquet-Bloch theory that the spectrum consists of a countable infinity of closed intervals, so it is ``banded", and that the corresponding solutions, the so-called Bloch waves, have the form
\begin{align*}
  u(x,t) = e^{i\omega_n(l) t} \, e^{ilx} \, \psi_n(x;l) \, , \quad n \in \mathbb{N}\, , l \in \left(-\pi, \pi\right] \, ,
\end{align*}
with $\psi_n(x;l)$ a periodic function in $x$ with the same period as the underlying coefficients and $\left(\omega_n(l)\right)_{n \in \mathbb{N}}$ a collection of spectral bands fulfilling
$$\omega_n(l) \leq \omega_{n+1}(l), l \in \left(\pi, \pi\right] \,.$$
Using Floquet theory, one can determine the band structure $(l, \omega_n(l))$ via
\begin{align}\label{eq:Floquet_stuff}
    e^{ \pm il} = \frac12 \left( \mathcal{D}(\omega) \pm \sqrt{\mathcal{D}(\omega)^2 - 4} \right) \, ,
\end{align}
so $\pm il$ are the Floquet exponents and
\begin{align}\label{eq:Floquet_discriminant_general}
\mathcal{D}(\omega) = {\rm trace}(\Phi_{\omega}(x)|_{x=1})
\end{align}
is the Floquet discriminant with
\begin{align}\label{EQ:bandstructure_ODE}
\frac{d}{dx} \Phi_{\omega}(x) = \left(
    \begin{array}{cc}
     0 & 1 \\ -(q(x)+s(x)\omega^2) & 0
    \end{array}
    \right) \Phi_{\omega}(x)\, , \quad \Phi_{\omega}(0) = Id \, .
\end{align}
In particular, any $\omega$ with $|\mathcal{D}(\omega)| > 2$ cannot fulfill \eqref{eq:Floquet_discriminant_general} for $l \in \mathbb{R}$, so it must necessarily fall into a spectral gap. This connection between band structure and fundamental system explains why it is either technically involved or impossible to get an explicit expression for the spectral bands, the bottleneck being explicit, workable expressions for a fundamental system for \eqref{EQ:bandstructure_ODE}. Here we focus on the special case of step function potentials $s = s_{\rm step}$ as in \eqref{EQ:s_step} (see Figure~\ref{FIG:s_coeff}).


\begin{Lemma}[\underline{Exact band structure}]\label{lemma:exact_band_structure}
Consider \eqref{eq:linear_problem} with $q(x) = 0$ and $s, \omega$ defined as follows:
\begin{enumerate}
\item For $p \in (3/8,1/2)$, let $s(x+1) = s(x) \, , x \in \mathbb{R},$ and
\begin{align}\label{EQ:s_step}
  s(x) = s_{step}(x) := 1 + C \chi_{[p,1-p)}(x) \, , \ x \in [0,1) \, , \sqrt{1+C} = \frac{2p}{3(1-2p)} \, .
\end{align}
\item Let 
\begin{align}\label{eq:bragg}
  \omega = \omega_*(s) = \frac{\pi}{\int_0^1 \sqrt{s_{step}(\tau)} \, d\tau } =  \frac{3\pi}{8p} \quad \mathrm{(``Bragg \ frequency")}.
\end{align}
\end{enumerate}
Then it holds true that
\begin{align}\label{eq:floquet_discriminant}
   \mathcal{D}(\omega) :=  -\frac{1}{12 p(2p-1)} \left( (3-4p)^2 \cos\left[ \frac83 p \sqrt{\omega^2}\right] - (3-8p)^2 \cos\left[\frac43 p \sqrt{\omega^2}\right]\right) \, ,
 \end{align}
such that the band structure $(l, \omega_n(l))$ can be obtained explicitly using \eqref{eq:Floquet_stuff}. Furthermore, we have $\mathcal{D}(m \omega_*;s) < -2 \, , \quad m \in \mathbb{N}_{odd}$
and so $\omega = m \omega_*,  m \in \mathbb{N}_{odd},$ cannot fulfill the band structure relation \eqref{eq:Floquet_stuff}, meaning that any odd multiple of the Bragg frequency is located in an odd spectral gap.
\end{Lemma}

\begin{proof} Observe that
 \begin{align*}
  Y'(x) = \left( \begin{array}{cc} 0 & 1\\ -\lambda r  & 0  \end{array} \right)  Y(x) \, ,
 \end{align*}
 has a fundamental matrix
 \begin{align*}
 \Phi(x) = \left( \begin{array}{cc} \cos(\sqrt{\lambda r} \, x) & \frac{1}{\sqrt{\lambda r}} \sin(\sqrt{\lambda r}\, x)\\[.2cm] - \sqrt{\lambda r} \sin(\sqrt{\lambda r}\, x) & \cos(\sqrt{\lambda r}\, x)   \end{array} \right) \, .
 \end{align*}
 Hence, choosing $s = s_{\rm step}$ to be a step-function as in \eqref{EQ:s_step}
 \begin{align*}
   \mathcal{D}(\lambda;s_{\rm step}) = \mathrm{trace} \left[ \Phi(p) \Phi(1-2p) \Phi(p) \right] \, ,
 \end{align*}
 and a lengthy computation yields
 \begin{align*}
   \mathcal{D}(\lambda;s_{\rm step}) &=  \frac12 \left(2 + \frac{2+C}{\sqrt{1+C}} \right) \cos[u + v] - \frac12 \left(-2 + \frac{2+C}{\sqrt{1+C}} \right) \cos[u - v] \, ,
 \end{align*}
 with
 \begin{align*}
  u \pm v = \sqrt{\lambda} (2p \pm \sqrt{1+C} (1-2p)) \, .
 \end{align*}
 Furthermore, using the Bragg-frequency $\omega_*$ from \eqref{eq:bragg}, we set
 \begin{align*}
   \sqrt{\lambda} = m \omega_* = \frac{m \pi }{\int_0^1 \sqrt{s(\tau)} \, d \tau} = \frac{m \pi}{2 p + \sqrt{1+C}(1-2p)}  \, ,
 \end{align*}
 which gives
 \begin{align*}
  \cos( u + v ) = \cos(m \pi) = -1 \, , \quad  \cos( u - v ) = \cos\left(\frac12 m \pi\right) = 0 \, , \quad m \in \mathbb{N}_{odd} \, ,
 \end{align*}
 so
 \begin{align*}
   \mathcal{D}(m \omega_*;s_{\rm step}) =  - \frac12 \left(2 + \frac{2+C}{\sqrt{1+C}} \right) < - 2 \, , \quad m \in \mathbb{N}_{odd}  \, ,
 \end{align*}
 which is precisely what is stated in Lemma~\ref{lemma:exact_band_structure}.
 \end{proof}

\begin{figure}[!h]
\centering
\scalebox{0.4}{\includegraphics{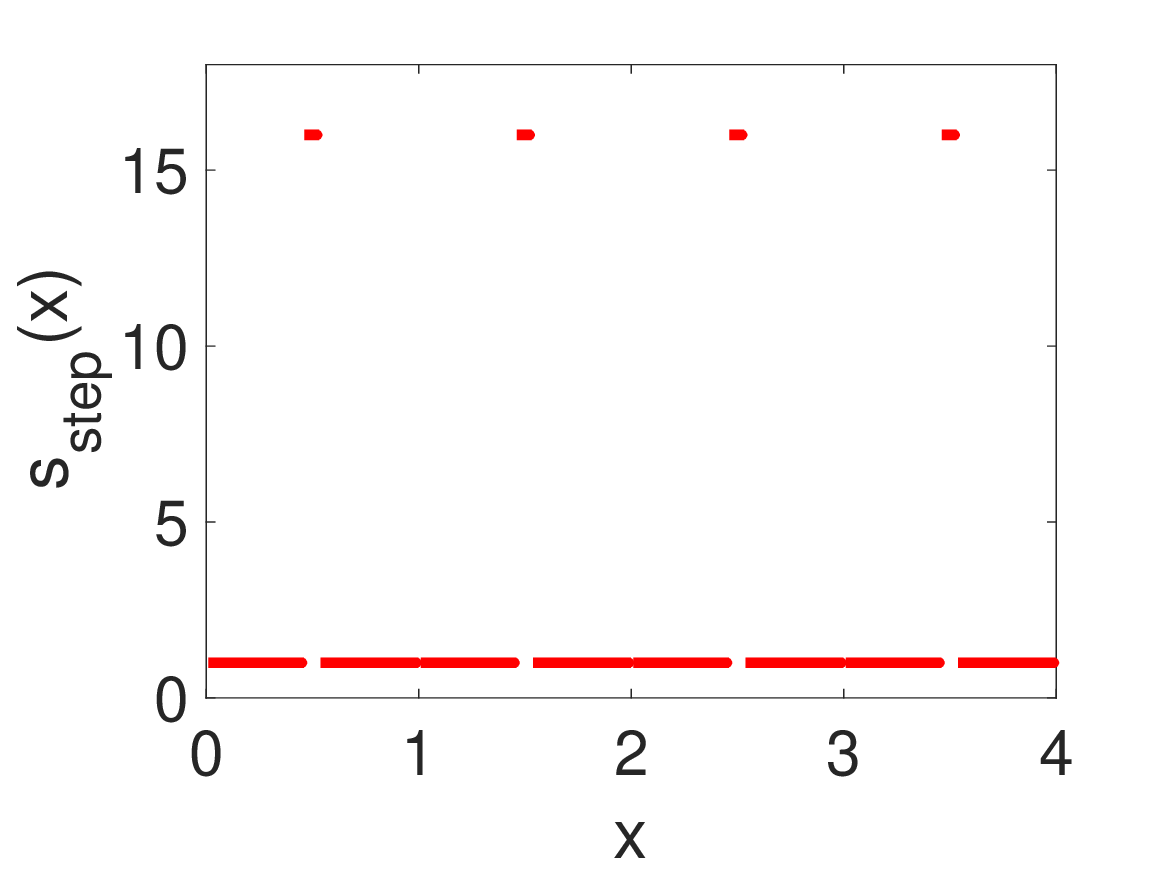}}
\scalebox{0.4}{\includegraphics{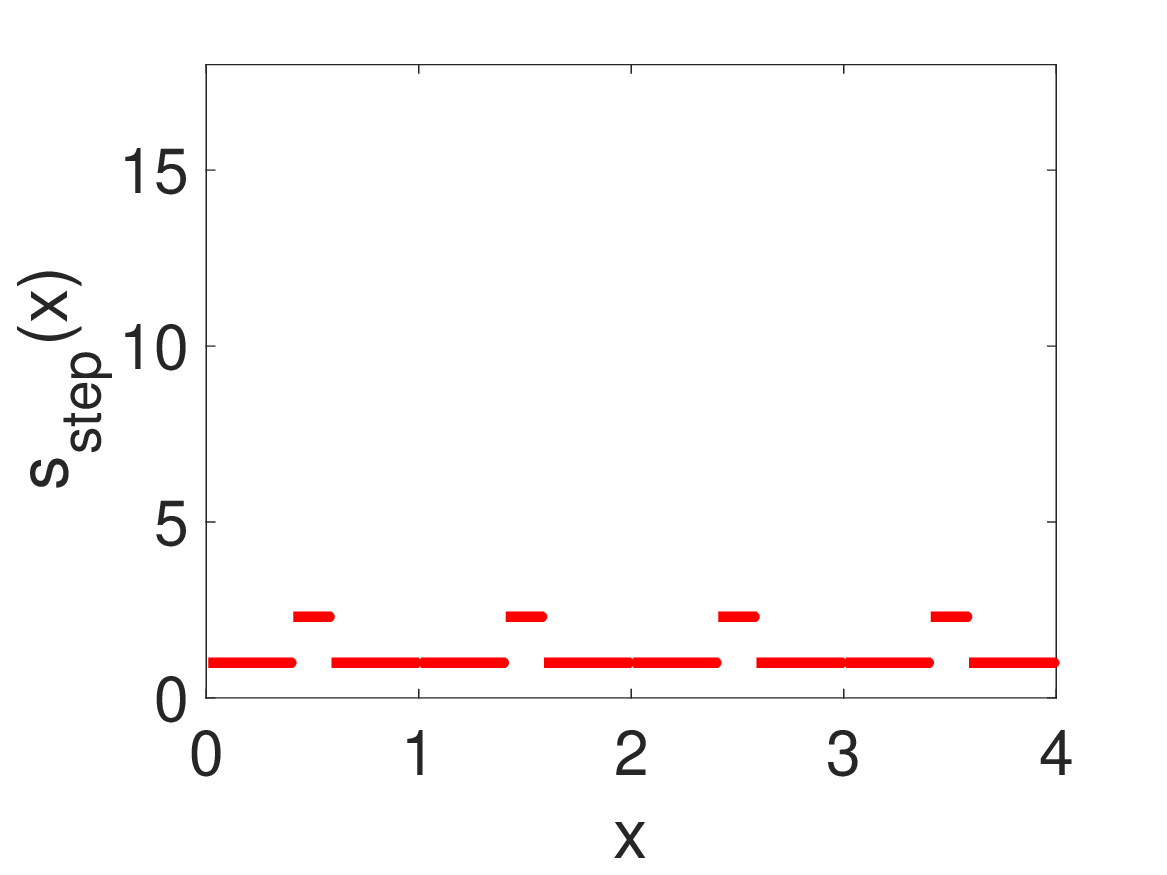}}
\caption{Coefficients $ s_{step} $ as defined in \eqref{EQ:s_step} for different values of $p$ (left: $p = 6/13\approx 0.4615$, right: $p = 0.41$). The closer $p$ is to $1/2$ the steeper and thinner the region unequal one. In the limit $p \rightarrow 3/8$ the step flattens and widens approaching the function identical 1. The steeper the step, the wider the band gaps. Numerically, it is more convenient to smooth out the steps using scaled versions of ${\rm tanh(x)}$.}
\label{FIG:s_coeff}
\end{figure}


\begin{figure}[!h]
     \centering
     \scalebox{1}{
     \begin{subfigure}[b]{0.5\textwidth}
         \centering
         \includegraphics[width=\textwidth]{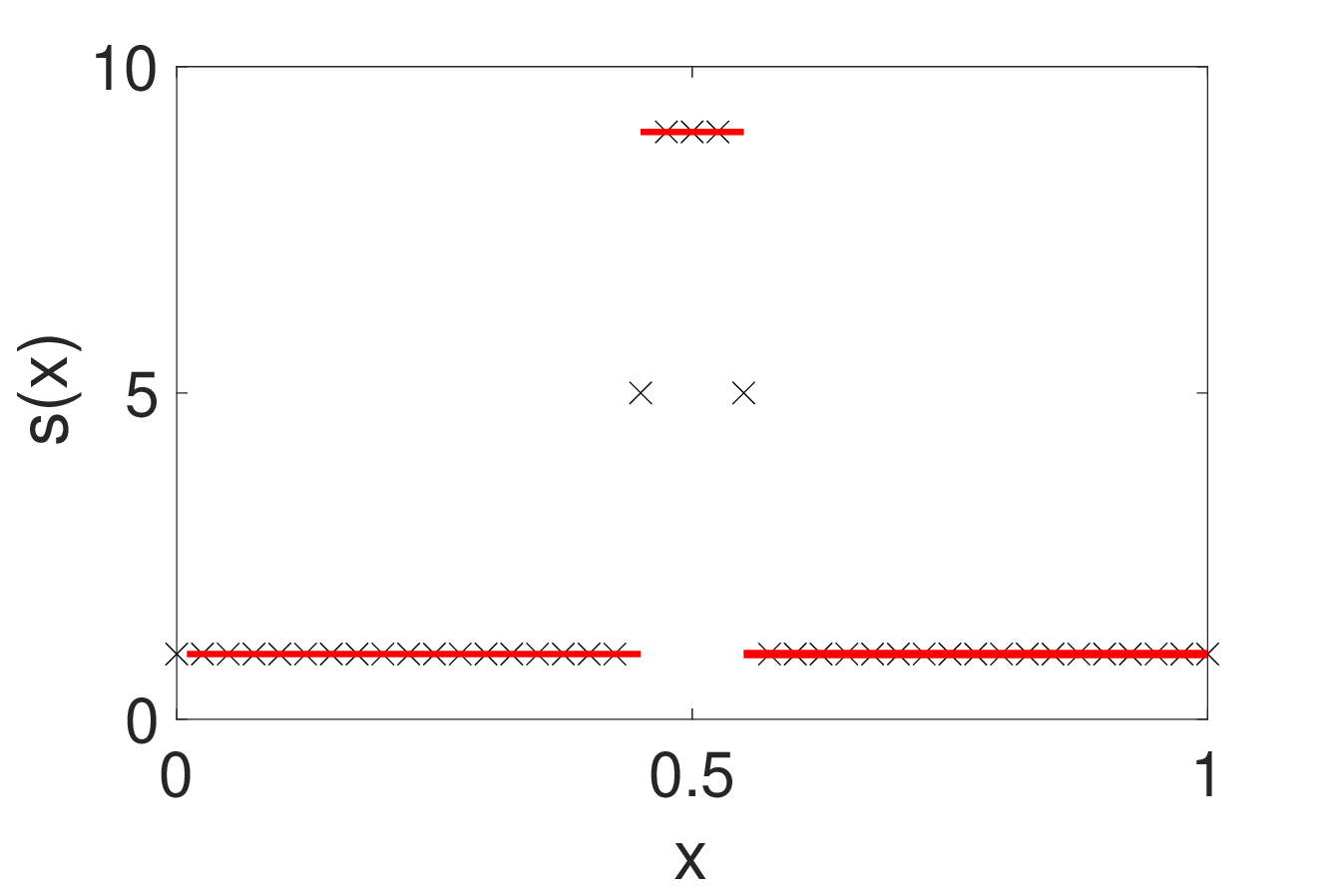}
         \caption{Exact step potential \eqref{EQ:s_step} {\it (solid red line)} and
           discretized version of the step function on the computational
           grid based on Eq.~\eqref{eq:sx_tanh} {\it (black crosses)}.}
         \label{fig:step_function_approximation}
     \end{subfigure}
     \hfill
     \begin{subfigure}[b]{0.45\textwidth}
         \centering
         \includegraphics[width=\textwidth]{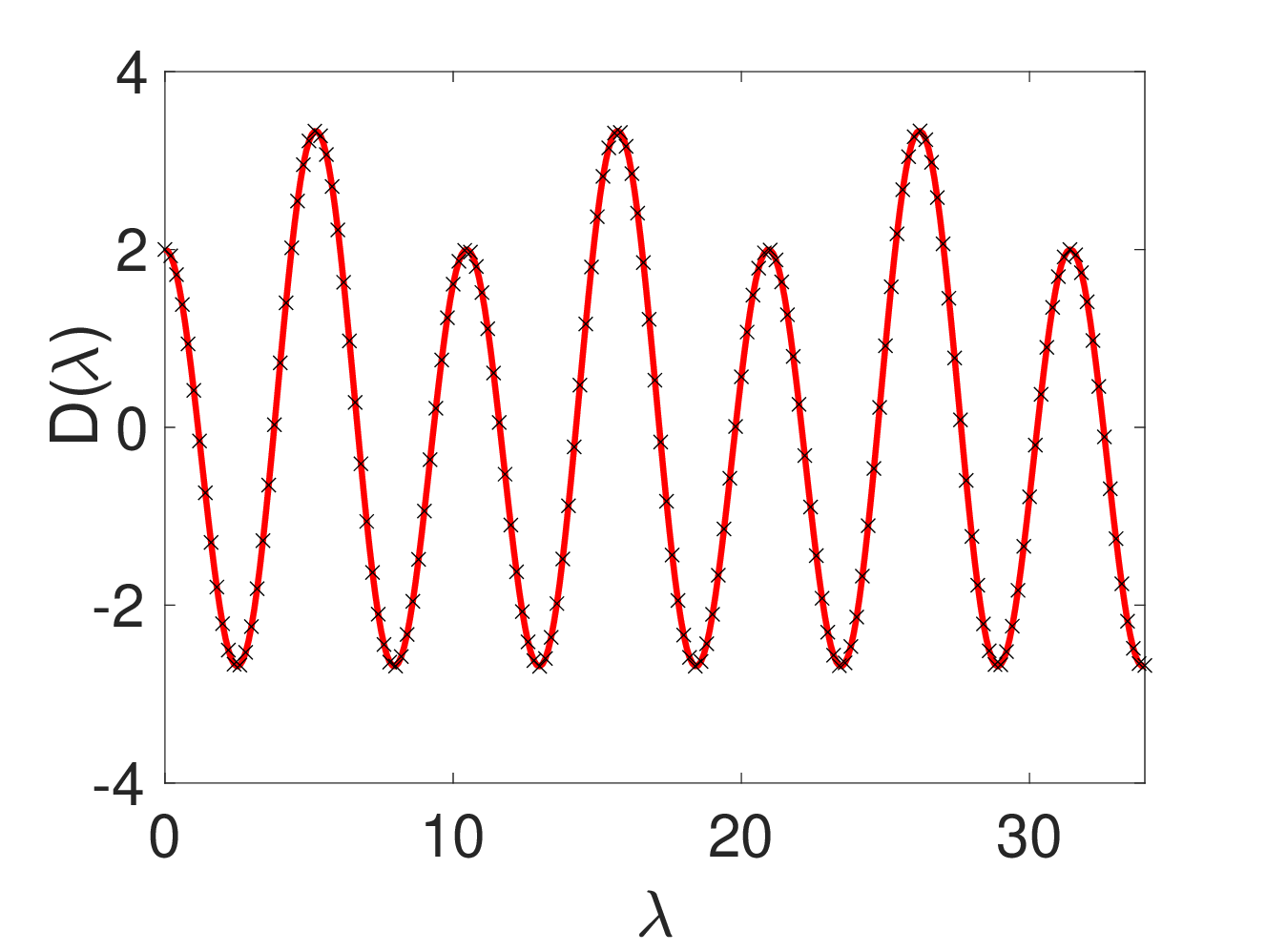}
         \caption{Exact discriminant \eqref{eq:floquet_discriminant} {\it (solid red line)} and numerical approximation of \eqref{eq:Floquet_discriminant_general} using \eqref{eq:sx_tanh} {\it(black crosses)}}
         \label{fig:discriminant}
     \end{subfigure}
     }
        \caption{Comparison of exact computation and numerical approximation of
          the step function heterogeneity and of the discriminant of Eq.~\eqref{eq:floquet_discriminant} (for $p = 0.45$) using Matlab (ODE15s).}
        \label{fig:step_function_and_discriminant}
\end{figure}

 The upshot of this result is that the coefficient $s$ and the breather frequency $\omega$ can be tuned to pairs $(s_{\rm step}, \omega_*)$ resulting in a band structure with uniformly open odd numbered band gaps into which multiples of the breather frequency $\omega_*$ fall. This is the main driving force of the breather construction: designing the spectral properties of the linear part to avoid resonances.
Note that the specific coefficients from \cite{BCLS11} correspond to $p = 6/13$ and have jump regions that are rather steep and thin, a feature which can be avoided by using a value of $ p < 6/13 $ (as is shown,
e.g., in the right panel of Figure~\ref{FIG:s_coeff}).\\

So far, the coefficient $q$ in \eqref{EQ:NLW} has been set to zero. Its role is to create a bifurcation of small-amplitude breathers from a band edge. The direct computation leading to \eqref{eq:floquet_discriminant} is difficult to carry out for non-zero $q$. We instead make use of a theoretical result that the band structure for the case $q = 0$ and $q \neq 0$ (with the same $s$ in both cases) only differs in the lower bands if $q \in L^2(0,1)$ (see, e.g. \cite{CW15}). This can be illustrated numerically (see Figure~\ref{fig:band_structure2} and Section~\ref{sec:num_band} for an explanation of the numerical method).


\begin{Definition}[\underline{Resonance-free triplet}]\label{def:resonance_free_triplet}
We call $(s_{\rm step}, q_{\rm step}, \omega_*)$ a ``resonance-free triplet'' if $s_{\rm step}, \omega_*$ are as in Lemma~\ref{lemma:exact_band_structure} and $ q_{step}(x):= s_{step}(x)(q_0 - \varepsilon^2)$ for $ \varepsilon > 0 $ with $q_0 \in \mathbb{R} $ chosen such that for the system
\begin{align}\label{EQ:onset}
\frac{d}{dx} \Phi_{\omega_*}(x) = \left(
    \begin{array}{cc}
     0 & 1 \\ s_{step}(x)(q_0 -\omega_*^2) & 0
    \end{array}
    \right) \Phi_{\omega_*}(x)\, , \quad \Phi_{\omega_*}(0) = Id \, ,
\end{align}
the Floquet discriminant $ \mathcal{D}(\omega_*; s_{\rm step}, q_0, \varepsilon)$ fulfills the following: 
\begin{itemize}
    \item At $\varepsilon = 0$, we have $ \mathcal{D}(\omega_*; s_{\rm step}, q_0, 0) =  -2, \mathcal{D}(m\omega_*; s_{\rm step}, q_0, 0) <  -2 $ for $ 1 < m \in \mathbb{N}_{\rm odd}$.
    \item For $\varepsilon > 0$ sufficiently small, we have $ \mathcal{D}(\omega_*; s_{\rm step}, q_0, \varepsilon) <  -2 $ for $ m \in \mathbb{N}_{\rm odd}$.
\end{itemize}
\end{Definition}

The main idea behind introducing $q$ in such a special way is to first determine $q_0$ (in terms of $s_{\rm step}$ and $\omega_{*}$) such that $\omega_*$ resides precisely on a band edge and then introduce a small bifurcation parameter $\varepsilon > 0$ that pushes the first band edge down in such a way that $\omega_*$ slips into a spectral band gap. The breather is predicted to exist for $\varepsilon > 0$ small enough.


\begin{figure}[!h]
\centering
\begin{minipage}{3cm}
 \scalebox{.42}{\includegraphics{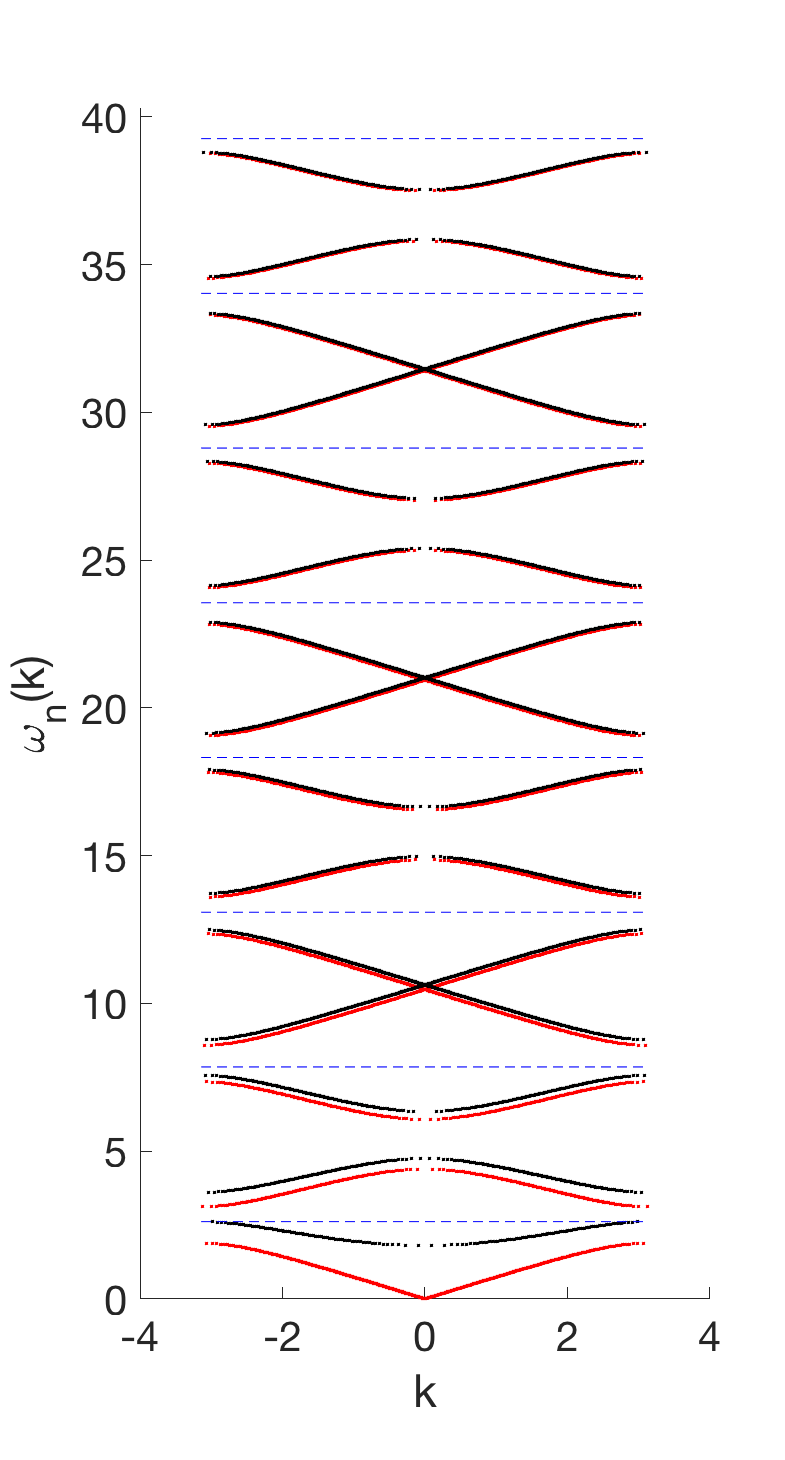}}
\end{minipage}
\hspace{3cm}
\begin{minipage}{5cm}
 \caption{Exact band structure computed from the exact discriminant \eqref{eq:floquet_discriminant} {\it (red)} for $s = s_{\rm step}, q = 0$ vs. numerically computed band structure for $s = s_{\rm step}, q = q_0s_{\rm step}$ with smoothed step function as in \eqref{eq:sx_tanh} using Matlab {\it (black)} along with odd multiples of $\omega_*$ from \eqref{eq:bragg} {\it (blue dashed lines)}. Parameter setting: $p = 0.45$. Observe how the choice of $q_0$ puts the first band edge at $\omega_*$ and how the band structure eventually approaches the exact one with $q = 0$ for higher bands, just as predicted by theory. For $s = s_{\rm step}, q = s(q_0-\varepsilon^2)$ with small $\varepsilon$ the frequency $\omega_*$ will move slightly into a spectral gap.}
 \label{fig:band_structure2}
 \end{minipage}
\end{figure}

\subsection{Existence result and its numerical implementation}\label{sec:num_band}

\begin{Proposition}[\underline{Existence of breathers}]\label{prop:breathers}
Let $ \rho < 0$ and let $(s, q, \omega) = (s_{\rm step}, q_{\rm step}, \omega_*)$ be a resonance-free triplet as in Definition~\ref{def:resonance_free_triplet}. Then there is $\varepsilon_0 > 0$ such that for all $ \varepsilon \in (0, \varepsilon_0) $ there exists a breather solution of \eqref{EQ:NLW} which can be approximated by
\begin{align}\label{EQ:breather_approximation}
 u_\mathrm{breather}(x,t) =   \varepsilon \eta_1 \, \mathrm{sech} \left( \eta_2 \varepsilon x \right) q_{11}(x)   \sin\left( \omega_* (t- t_0) \right)+ h.o.t. \, , \quad t_0 \in \mathbb{R} \, ,
\end{align}
with
\begin{align}\label{eq:constants}
 \eta_1 & := 2\sqrt{\frac{2\bar{s}_1}{\bar{s}_3}} \, ,
 \eta_2 := \sqrt{\bar{s}_1} \, , \\
 \bar{s}_1 & = - \left(\frac{1}{2 \det{S}} \right)\int_0^2 s_{\rm step}(x) q_{11}(x)^2 \, dx \, ,
 \bar{s}_3 = \frac{3 \rho}{2 \det{S}} \int_0^2 q_{11}(x)^4 \, dx \, , \nonumber
\end{align}
where $ Q(x) = (q_{jl}(x))_{j,l = 1,2} $ and $S$ are defined as follows: Let $ \Phi_{\omega_*} $ be a matrix solution for \eqref{EQ:onset}
and define $ Q(x) := P(x) S$ by the Floquet representation
\begin{align}\label{eq:floquet_representation}
 \Phi_{\omega_*}(x) = P(x) e^{x M} = P(x) S e^{x J}  S^{-1} \, , \quad J= \left(
    \begin{array}{cc}
     0 & 1 \\ 0 & 0
    \end{array}
    \right) \, ,  \quad P(x+2) = P(x) \, .
\end{align}
\end{Proposition}

{\bf Proof sketch.} In a first step the time-dependence is eliminated by representing $u$ by its Fourier series in time
\begin{align*}
    u(x,t) = \sum_{n \in \mathbb{Z}} \hat{u}_n(x)e^{in \omega t}
\end{align*}
which turns the PDE into a countably infinite nonlinear system 
\begin{align*}
    0 = \hat{u}_n^{\prime \prime} + \left(\omega^2 s(x)+q(x)\right) \hat{u}_n + g_n((\hat{u}_k)_{k \in \mathbb{Z}}) \,, \quad n \in \mathbb{Z} \,,
\end{align*}
for the (temporal) Fourier coefficients $\hat{u}_n$ and a convolution nonlinearity $g_n$. Upon restricting to an invariant subspace (arising through various symmetries that \eqref{EQ:NLW} is naturally equipped with) one is left with a similar countably infinite system of ODEs, but now for a real-valued variable $u_n, n \in \mathbb{N}_{\rm odd}$. The tuning of the band structure via the resonance-free triplet from Definition~\ref{def:resonance_free_triplet} results in a Floquet exponent configuration that allows for the use of a blend of center manifold reduction for PDEs and bifurcation theory (see, e.g. \cite{HI11}) to reduce the analysis to a planar ODE with periodic coefficients. This reduced ODE is treated via averaging methods and finally demonstrated to support homoclinic (to zero) orbits (explicitly given in terms of $\mathrm{sech}$-functions). These then finally give the desired breather solution. See \cite{BCLS11} for details.\\

{\bf Remark.} It turned out to be computationally convenient in the breather construction to take $s, q$ to be $2$- instead of $1$-periodic, since in that case the Floquet exponents turned out to be purely real. We do not further comment here on that (see \cite{BCLS11}).\\

{\bf Related work.} Note that breather solutions are closely related to standing and traveling modulated pulses/wave packets whose existence (on long, but finite time scales) is usually demonstrated via derivation and justification of the Nonlinear Schr\"odinger equation as amplitude equation (see \cite{BSTU06} for a single wave packet and \cite{CS12} for wave packet interaction). Furthermore, there are so-called generalized breather solutions which feature (small) periodic tails (see \cite{LBCCS09}) due to resonances with the band structure. The main motivation for studying nonlinear wave equations with spatially periodic coefficients comes from an application in photonics where one  strives to design periodic materials (represented by the periodic coefficients) such that there exist standing light pulses (represented by breather solutions) (see \cite{BT04}).\\

{\bf Linear stability analysis.} Linearizing around the breather solution \eqref{EQ:breather_approximation} results in a linear operator
\begin{align}\label{eq:Floquet_operator_breather}
    \mathcal{L}_{\rm breather}
    \begin{pmatrix}
    u \\ v
    \end{pmatrix}
    := \left( \frac{1}{s_{\rm step}(x)} \right)\begin{pmatrix}
    0 & 1\\ \partial_x^2 + q_{\rm step}(x) + 3 \rho u_{\rm breather}(x,t)^2 & 0
    \end{pmatrix}
    \begin{pmatrix}
    u \\ v
    \end{pmatrix}\, .
\end{align}
The time periodicity of the breather will result in the use of Floquet theory in time. Due to the temporal translation invariance of the PDE, one immediately has that $[\partial_t u_{\rm breather}(x,0),\allowbreak \partial_t^2 u_{\rm breather}(x,0)]$ is an eigenfunction belonging to the Floquet multiplier $+1$ (or Floquet exponent $0$). This eigenfunction is associated with the so-called phase mode. Furthermore, the essential spectrum is related to the band structure of the linear PDE \eqref{eq:linear_problem}. Retrieving any further information on the spectrum analytically is difficult. Hence, in the remaining part of the manuscript, we will turn to numerical methods to analyze the stability
of the constructed breather structures.\\

{\bf Numerical computation of $q_0$ from Definition~\ref{def:resonance_free_triplet}.} After choosing $ p \in (3/8, 1/2) $, one can compute $s_{\rm step}$ and thereby $\omega_*$ as in Lemma~\ref{lemma:exact_band_structure}. This then allows to compute $q_0$ by numerically defining a function $ T(\mu) := \mathrm{trace} \left( \Phi_{\omega_*}(x; s_{\rm step}, \mu, 0)|_{x = 2} \right) +2 \ ,
$ with $\Phi_{\omega_*}$ the canonical fundamental matrix corresponding to \eqref{EQ:onset}. We have suppressed its dependence on $s_{\rm step},\omega_*$ since these are kept fixed, while we vary $q_0$ in $q_{\rm step}(x) = s_{\rm step}(x)q_0$. Numerical root finding for $T(q_0)=0$
gives then $q_0$ (upon having a good enough first guess). Note that there is also a way of finding an explicit formula for $q_0$ (see \cite{BCLS11}) in the step function case, but this is of limited numerical use since it will turn out to be more convenient to perform numerical computations for smoothed step functions.\\

{\bf Numerical band structure computation for the linear problem.}
Let us now also briefly discuss how to numerically implement the computation of the band structure for the linear PDE \eqref{eq:linear_problem}. Using Floquet theory, one can compute for which values of $\omega$ there exists a bounded solution $v$ of \eqref{eq:spectral_problem} by inspecting the Floquet exponents (which determine if a fundamental system for \eqref{eq:spectral_problem} gives bounded solutions or exponential growth/decay). In particular, one has a relation between the Floquet exponents  $\widetilde{l}_{\pm} = \widetilde{l}_{\pm}(\omega) \in \mathbb{C}$  and the spectral parameter $\omega$ given by
\begin{align}\label{eq:Floquet_stuff_2}
    e^{ \widetilde{l}_{\pm}} = \frac12 \left( \mathcal{D}(\omega) \pm \sqrt{\mathcal{D}(\omega)^2 - 4} \right) \, .
\end{align}
Moreover, for our equation we have that the Floquet exponents are either purely real or purely imaginary (with the usual non-uniqueness from possibly adding multiples of $2\pi i$ due to the periodicity of the complex exponential). One can now choose values $\omega \in \mathbb{R}$ and numerically compute for each such $\omega$: (1) the canonical fundamental matrix $\Phi_{\omega}$ for \eqref{EQ:bandstructure_ODE}, (2) the Floquet exponents from \eqref{eq:Floquet_stuff_2}, and finally (3) if $|\mathcal{D}(\omega)| < 2$, then $\widetilde{l}_{\pm} = \pm i l \in i \mathbb{R}$ and $\omega$ belongs to a spectral band. Plotting the relation $\pm l$ vs $\omega$ (for $\omega$ in a spectral band) gives the band structure, while simply plotting $\omega$ belonging to a spectral band gives the "banded" spectrum. See Figure~\ref{fig:band_structure2} for a numerical computation of the band structure for $q(x) = q_0 s_{\rm step}(x)$.\\

{\bf Numerical computation of the breather approximation $u_{\rm breather}$ from \eqref{EQ:breather_approximation}.} Equipped with a\\ resonance-free triplet $(s_{\rm step}, q_{\rm step}, \omega_*)$, one can numerically compute the canonical fundamental matrix $\Phi_{\omega_*}(x; s_{\rm step}, q_0, 0)$. Using Floquet theory, one can derive the representation \eqref{eq:floquet_representation}, which, in particular, gives $\Phi_{\omega_*}(2) = e^{2 M}$ such that $M$ can be determined by numerically computing the matrix logarithm. Note that one actually expects a Jordan block to arise (see \cite{BCLS11} for details), so computing $S$ and $J$ numerically can be a delicate task. Finally, setting $P(x) = \Phi_{\omega_*}(x) e^{-x M}$ and $Q(x) = P(x)S$ enables the computation of all necessary constants in \eqref{eq:constants} by numerical integration. See Figure~\ref{FIG:breather} for an illustration of such a computation.

\section{Stability of breathers}\label{subsec:stability_result}

We now consider the above mentioned setup
of a resonance-free triplet \eqref{def:resonance_free_triplet}
numerically towards identifying a numerically
exact (up to a prescribed accuracy)
breather waveform.
Throughout this section, we fix $p=0.45$ and $\varepsilon=0.3$. This value of $\varepsilon$ is selected so that the breather is localized  for the chosen values of the domain length and $p$.
{More
concretely the breather's FWHM in this case 
is smaller than 1/8 of the domain length.}

In order to identify breather solutions
practically, we discretize the relevant
Klein-Gordon PDE. Among all the discretization schemes, we have chosen
for simplicity to utilize finite differences (but also confirmed their
ability to capture theoretically predicted features as will be discussed
below). To this aim, we choose a uniform grid whose lattice spacing is given by $h$. Due to the large domain needed for containing the breather, we need to choose a value of $h$ that gives a tractable lattice size. The domain extends over the interval $[-L,L)$ (with periodic boundary conditions), and the number of lattice sites is $N=2L/h$ (for simplicity $N$ is taken as an even integer number). In our numerics, we have taken $h=0.02$ and $L=60$ so that $N=6000$.
The value of $x$ at the lattice nodes, i.e. $x_n\equiv x(n)$, is dependent of the choice of $p$. For the numerics, we have taken $p$ with two decimal digits, and we have chosen $x_n=nh$ if $100p$ is even and $x_n=(n+1/2)h$ if $100p$ is odd.

With this discretization, we can write (\ref{EQ:NLW}) as
\begin{equation}\label{eq:DKG}
s_n\ddot{u}_n = \frac{1}{h^2}\Delta u_n - q_nu_n - u_n^3\,,\qquad n=-N/2\ldots N/2-1
\end{equation}
with $u_n\equiv u(x_n)$ and $q_n=s_n(q_0-\varepsilon^2)$. In order to get a good correlation between the analytical and the numerical spectrum for linear modes, a sixth-order discretization for $\partial_x^2$  is introduced, {\it i.e.} \cite{Fornberg},

\begin{equation}
\Delta u_n=\frac{1}{90}(u_{n-3}+u_{n+3})-\frac{3}{20}(u_{n-2}+u_{n+2})+\frac{3}{2}(u_{n-1}+u_{n+1})-\frac{49}{18}u_n.
\end{equation}

Numerically, from a practical perspective, we have found
the choice of $s_n$ to be central towards the convergence
of our numerical scheme.
Here, we define the function $s(x)$ so that just in the step (i.e., for $x=p$ or $x=1-p$), the value of $s(x)$ takes an intermediate value, i.e. $1+C/2$. That is,
\begin{equation}
s_n=1+C\times
\begin{cases}
0, & x_n<p \\
1/2, & x_n=p \\
1, & p<x_n<1-p\\
1/2, & x_n=1-p\\
0, & x_n>1-p\\
\end{cases}
\qquad \text{if}\ x_n\in[0,1)
\end{equation}

For our numerical purposes, the choice of $s(x)$ that
we have found to correlate efficiently with this requirement is an
approximated step function in the form:
\begin{equation}\label{eq:sx_tanh}
s(x)=1+\frac{C}{2}\left[\tanh(\mu(x-p))+\tanh(-\mu(x-(1-p)))\right]
\end{equation}
with a high value of $\mu$ such as $10^5$. This choice gives a value of $\omega^*=2.61778$, which is very close to the analytical one corresponding to the step function, namely $\omega^*=3\pi/(8p)=2.61800$. This choice also gives the value $q_0=3.2701$ for the constant appearing in Definition \ref{def:resonance_free_triplet}. We have compared the exact band structure {determined in Section \ref{sec:num_band}}, stemming from linearizing (\ref{EQ:NLW}), with the one arising from the linearization of (\ref{eq:DKG}). In other words, if one introduces the linear modes expression $u_n(t)=U_n\exp(i\omega_l t)$ at (\ref{eq:DKG})
around the trivial equilibrium $u_n=0$,
the following generalized eigenvalue problem is obtained:
\begin{equation}\label{eq:linear}
    \omega_l^2s_nU_n=\left(\frac{1}{h^2}\Delta-q_n\right) U_n
\end{equation}
whose numerical diagonalization gives the spectrum of linear modes $\omega_l$. Figure~\ref{fig:linearmodes} shows the analytical and numerical linear modes spectrum for $p=0.45$ and the odd integer multiples of $\omega^*$. One can see that there are resonances with the numerical spectrum for the 13th (and higher) harmonic. However, as we explain below, they have no effect in the breather existence.

\begin{figure}[!h]
\centering
\begin{minipage}[b]{3cm}
 \scalebox{.5}{\includegraphics{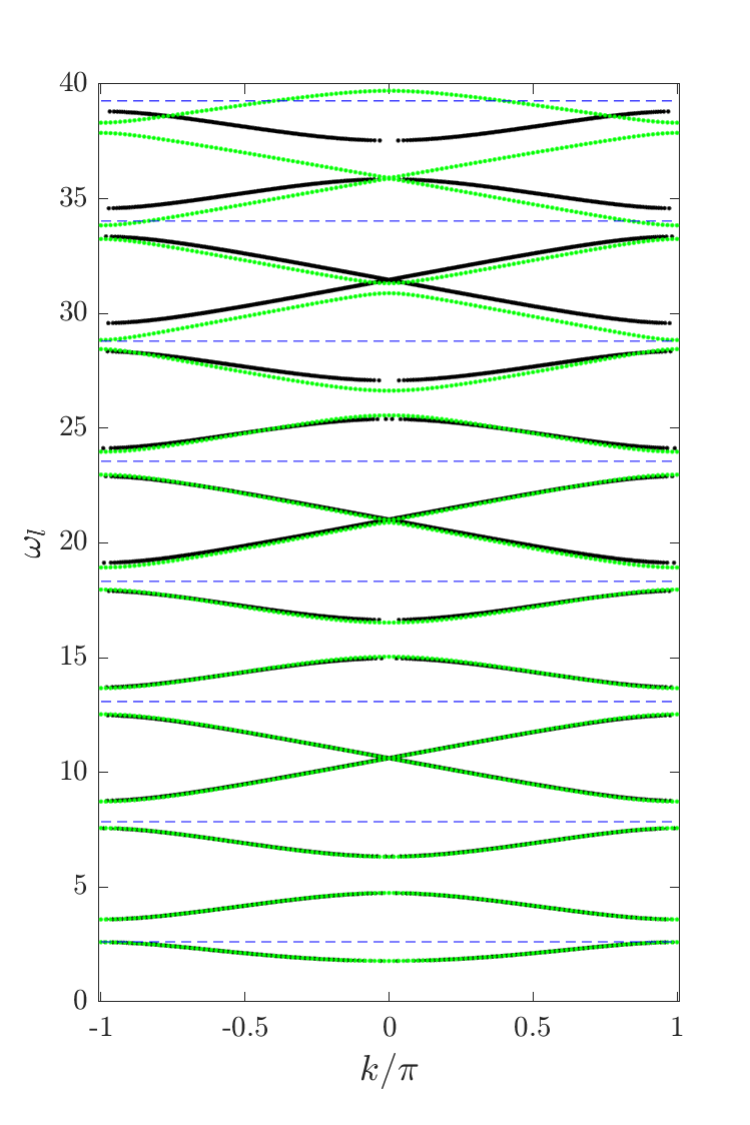}}
\end{minipage}
\hspace{3cm}
\begin{minipage}[b]{5cm}
 \caption{{Numerical linear dispersion for $p=0.45$ obtained from (\ref{eq:linear}) (green dots) and numerically computed dispersion relation via the Floquet discriminant (black dots) for $s(x)$ as in (\ref{eq:sx_tanh}). The frequency $\omega^*$ and its odd multiples up to 15 are indicated by blue dashed lines.}}
 \label{fig:linearmodes}
 \end{minipage}
\end{figure}

In order to get time-reversible breathers from (\ref{eq:DKG}),
one can work in the Fourier space by expanding $u_n(t)$
into associated modes according to:
\begin{equation}
    \label{eq:Fourier}
    u_n(t)=z_{0,n}+2\sum_{k=1}^{K}z_{k,n}\cos(k\omega t),\quad z_{k,n}\equiv z_k(x_n)
\end{equation}
Thus, (\ref{eq:DKG}) transforms into a set of $(K+1)N$ nonlinear algebraic equations $\mathbf{F}(\{z_{k,n}\})=0$:

\begin{equation}
    \label{eq:KG_Fourier}
    F_{k,n} \equiv -k^2 \omega^2 s_n z_{k,n}-\frac{1}{h^2}(z_{k,n+1}+z_{k,n-1}-2z_{k,n})+q_n z_{k,n}+ {\mathcal F}_{k,n}=0.
\end{equation}

Here, ${\mathcal F}_{k,n}$ denotes the $k$-th mode at the $n$-th site of the discrete cosine Fourier transform of $u_n^3$:

\begin{equation}
    {\mathcal F}_{k,n}=\frac{1}{2K+1}\left[u_n^3(0)+2\sum_{q=1}^K u_n^3(t_q)\cos(k\omega t_q)\right],
\end{equation}
with $u_n(t)$ taken from \eqref{eq:Fourier} and $t_q=2\pi q/((2K+1)\omega)$.

For solving (\ref{eq:KG_Fourier}), we make use of fixed point methods. Among those methods, we have chosen the trust-region-dogleg which is the default algorithm in Matlab's \texttt{fsolve} function. In order to implement the fixed point method, we have chosen the initial guess as $z_{k\neq1,n}=0,\ z_{1,n}=u_\mathrm{breather}(x_n,t)$
where $u_\mathrm{breather}$ refers to the analytical
breather approximation of Eq.~(\ref{EQ:breather_approximation}).
Figure~\ref{fig:profile} shows the profile of the breather for $\omega=\omega^*$ and $p=0.45$. This breather can be continued until a resonance with linear modes occurs. Taking into account that the breather possesses harmonics up to the $K$-th one
(within our Ansatz), only resonances of $k\omega^*$ with $k\leq K$ odd are considered. With those constraints, and accounting for
the fact that we have taken $K=11$, the {resonances with the 13th (and higher) harmonics observed in Fig.~\ref{fig:linearmodes} are irrelevant
  for our considerations herein.} For the choice $p=0.45$ there are no resonances in the interval $\omega\in(2.6014,2.6221)$. The lower bound within the interval corresponds to the first harmonic resonance with the top of the first band, whereas the upper bound holds for the resonance of the 11th harmonic with the bottom of the tenth band. However, the breather can be continued past the latter boundary as the effect of such resonance is introducing a modification of the breather of the order of $10^{-8}$. Similarly, the next resonance, which occurs for $\omega=2.6639$ and comes from the 9th harmonic, also has a negligible effect (of the order of $10^{-7}$). Figure~\ref{fig:energy} shows the energy versus $\omega$ for $p=0.45$ where one can see that the energy tends to zero at the lower bound of the interval (as the breather bifurcates from the corresponding $\pi$-phonon at the upper edge of the band  ). The energy of the breather is equivalent to the Hamiltonian associated with (\ref{EQ:NLW}), i.e.

\begin{equation}\label{eq:Hamiltonian}
    H=\int \left[\frac{1}{2}s(x)(\partial_t u(x,t))^2+\frac{1}{2}(\partial_x u(x,t))^2+\frac{1}{2}q(x)u^2(x,t)+\frac{1}{4}u^4(x,t)\right]\,\mathrm{d}x
\end{equation}

\begin{figure}[!h]
\begin{center}
\begin{tabular}{cc}
\includegraphics[width=.45\textwidth]{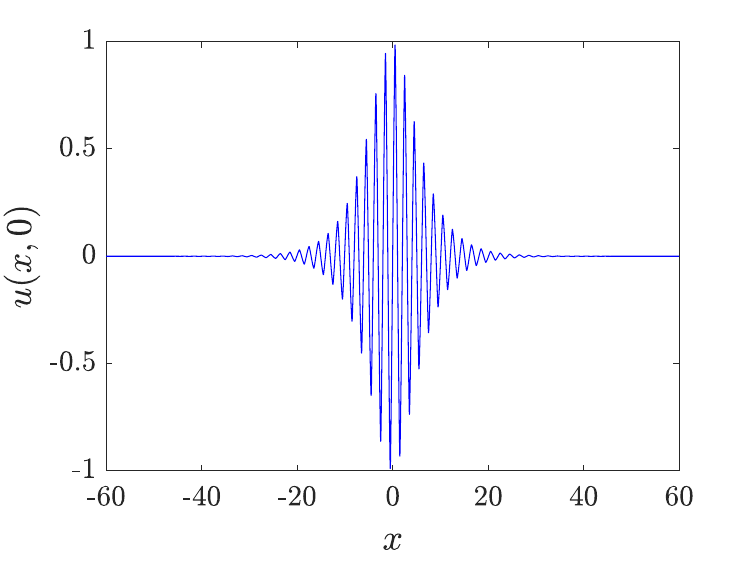} &
\includegraphics[width=.45\textwidth]{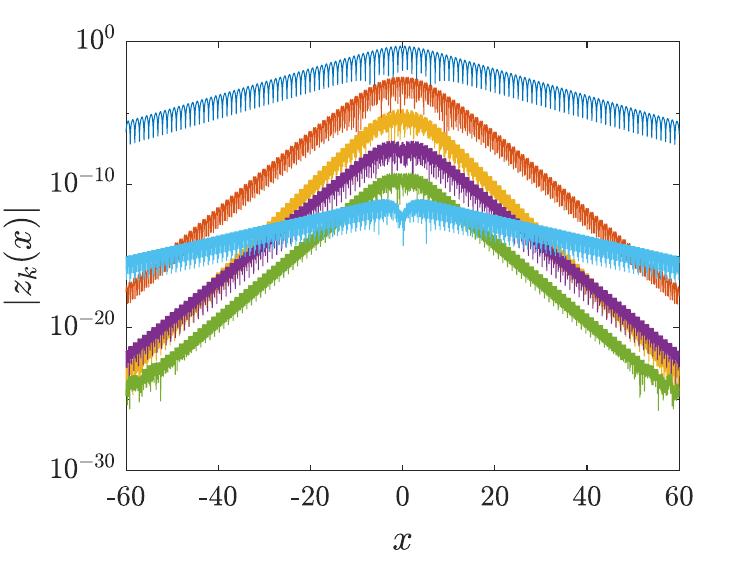} \\
\end{tabular}
\end{center}
\caption{Breather at $p=0.45$ and $\omega=\omega^*$. Left panel shows the profile at $t=0$ while the right panel shows the odd-k Fourier coefficients $|z_{k}(x)|$ at semilogarithmic scale.}
\label{fig:profile}
\end{figure}

\begin{figure}[!h]
\begin{center}
\includegraphics[width=.45\textwidth]{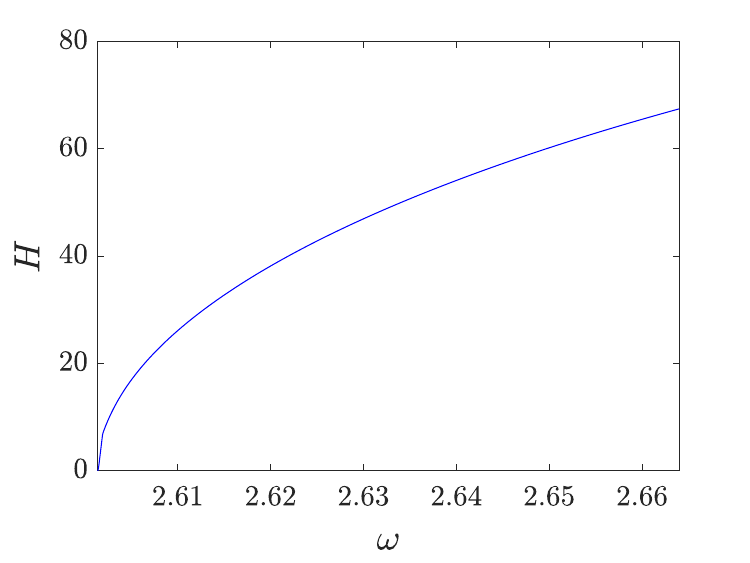}
\end{center}
\caption{Dependence of the energy as a function of the
breather frequency for $p=0.45$.}
\label{fig:energy}
\end{figure}

The stability properties of the obtained solutions are
identified by means of Floquet analysis. To that effect, we  add a perturbation $\xi(x,t)$ to the solution $u(x,t)$ at (\ref{EQ:NLW}). The resulting linearized PDE reads:

\begin{equation}\label{eq:perturb}
    s(x) \partial_t^2 \xi = \partial_x^2 \xi - (q(x) + u^2)\xi
\end{equation}

In order to perform the Floquet analysis, the aim is to compute the spectrum of the Floquet operator, whose matrix representation is known as the monodromy matrix $\mathcal{M}$. This is defined according to the map:

\begin{equation}
    \Omega(T)=\mathcal{M}\Omega(0),\qquad \Omega(t)=[\xi(x,t),\partial_t\xi(x,t)]
\end{equation}

The eigenvalues of $\mathcal{M}$ represent the Floquet multipliers, and can be written as: $\lambda=\exp(i\theta)$. Given the real, symplectic nature of the Floquet operator, the multipliers come in pairs $(\lambda,1/\lambda)$ if they are real, or in quadruplets $(\lambda,\lambda^*,1/\lambda,1/\lambda^*)$ if they are complex. For a periodic solution to be stable, generally, it is needed that $|\lambda|\leq1$. In our more specific Hamiltonian setting, stability necessitates that all the eigenvalues are on the unit circle. Additionally, there is always a pair of eigenvalues at $\theta=0$, given its Hamiltonian  nature.

More precisely, due to the invariance of
the model under consideration under time translation, there is an eigenvalue pair at $1+0i$. {As it is mentioned after (\ref{eq:Floquet_operator_breather}), its associated eigenmode, known as phase mode,} corresponds to $\Omega(t)=[\partial_tu(x,t),\partial^2_t u(x,t)]$. In order to find numerically the monodromy spectrum, Eq.~(\ref{eq:perturb}) must be discretized on the same grid
as the one where the solution was identified. Then,
we perform simulations of these linearization
equations for a period. To this aim, we have used the fourth-order explicit and symplectic Runge-Kutta-Nystr\"om method developed in \cite{calvo} with a time-step $\delta t=T/1500$, with $T=2\pi/\omega$. With this choice, we get, for the breather with frequency $\omega^*$ and $p=0.45$, that the mode associated with the above symmetry
is at $\theta\approx9\times10^{-7}$; this
serves as a benchmark for the accuracy of our Floquet multiplier
computations. Very close to $1+0i$, we can find a real localized mode at a distance $\approx5\times10^{-7}$ from $1+0i$, indicating an instability (although with an extremely small growth rate). Figure \ref{fig:modes} shows the shape of both modes. Notice that{, on the one hand} for the phase mode, $\xi(x,0)=0$ because of the time-reversibility of the breather{; on the other hand, one might think that the localized mode could be related to a translational mode $\xi(x,0)=\partial_x u(x,0)$, a
feature that can occur in discrete Klein-Gordon lattices~\cite{AC98}. However, as shown in Fig.~\ref{fig:modes}, this localized mode has a different shape than the translational one. In addition, the $\partial_t\xi(x,0)$ component of the localized mode is not zero. Although in the case shown in Fig.~\ref{fig:modes} it is tiny compared to the $\xi(x,0)$ component, it grows with $\omega$ and, for instance, it is only 8 times smaller when $\omega=2.66$}. This suggests that some velocity may be
imparted on the structure upon perturbation and may
accordingly lead to potential mobility; see also below
for the relevant dynamics.

\begin{figure}[!h]
\begin{center}
\begin{tabular}{cc}
\includegraphics[width=.45\textwidth]{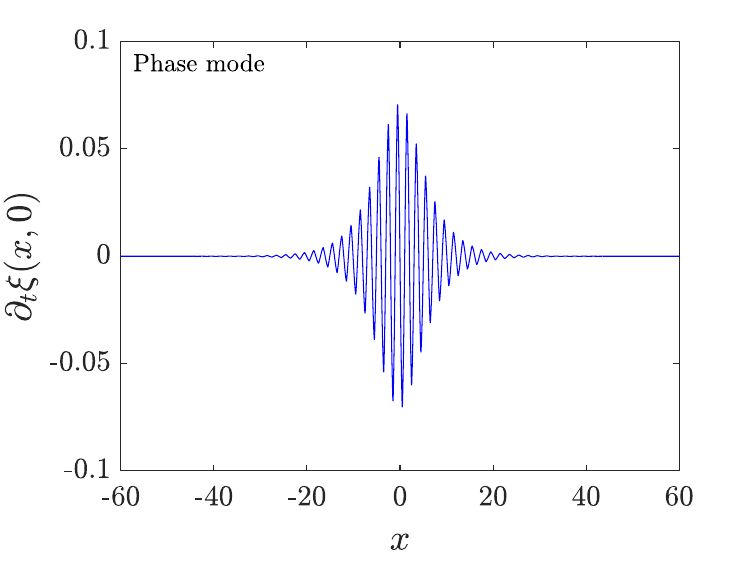} &
\includegraphics[width=.45\textwidth]{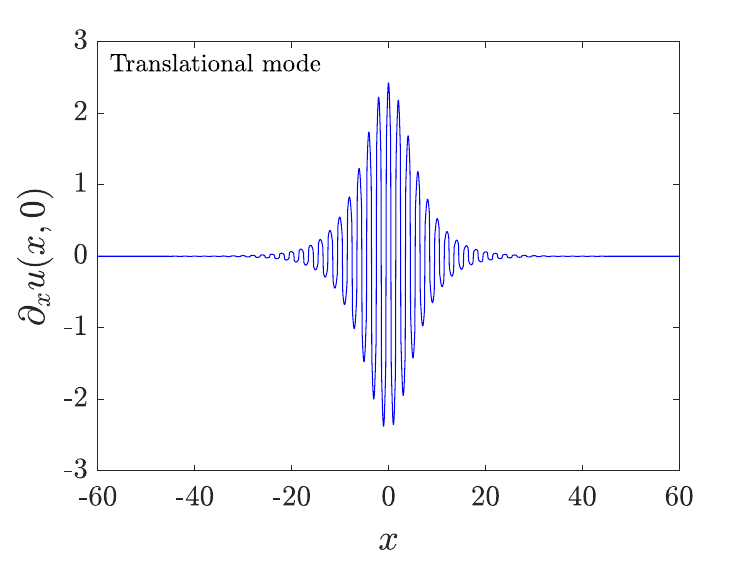} \\
\includegraphics[width=.45\textwidth]{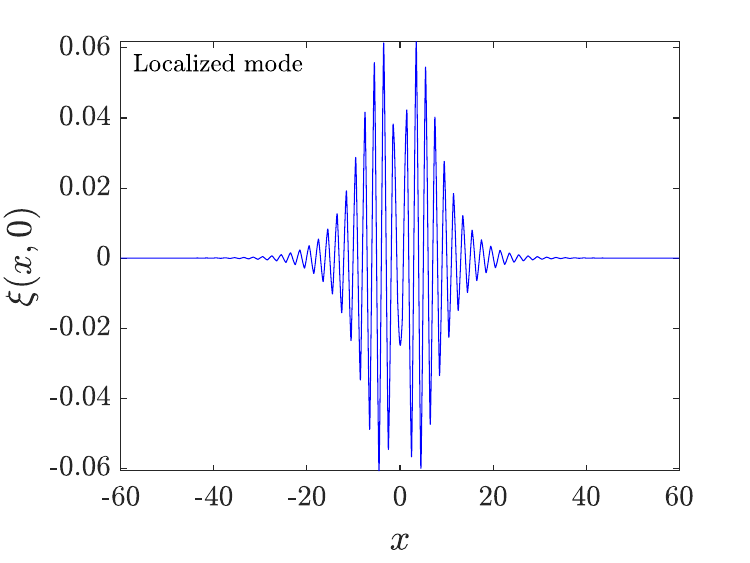} &
\includegraphics[width=.45\textwidth]{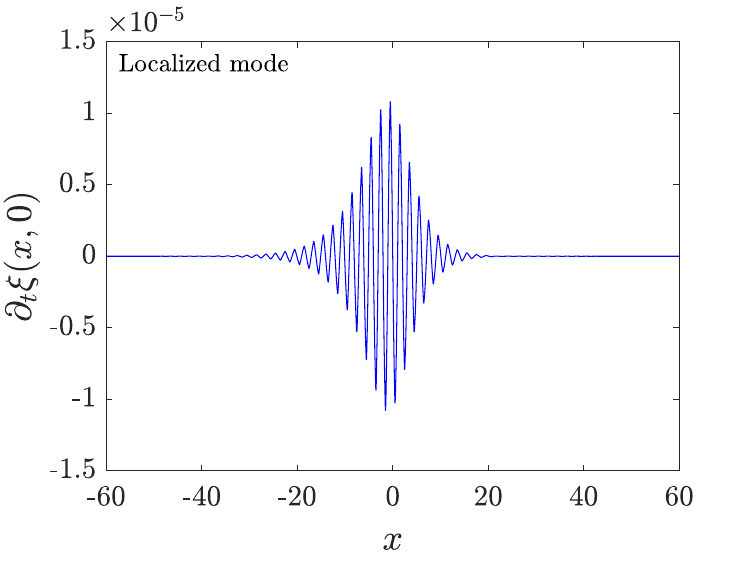} \\
\end{tabular}
\end{center}
\caption{{(Top left) Non-zero component of the phase mode. (Top right) Translational mode defined as $\partial_x u(x,0)$. (Bottom) Components of the localized mode. In every panel, $p=0.45$ and $\omega=\omega^*$.}}
\label{fig:modes}
\end{figure}

{When the frequency is varied from $\omega^*$, the localized mode is always real and higher than 1 and has a monotonically increasing behaviour with $\omega$. As $|\lambda|\geq1$ in all the considered interval, the breather will be unstable.}

\begin{figure}[!h]
\begin{center}
\includegraphics[width=.45\textwidth]{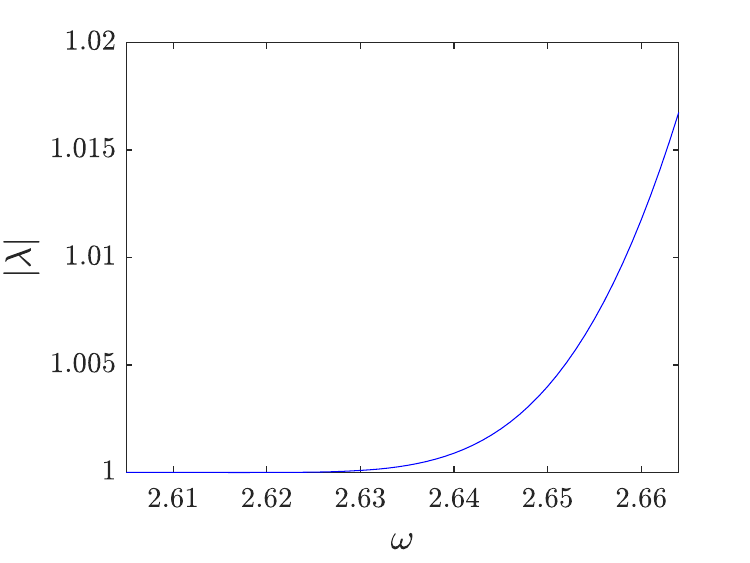}
\end{center}
\caption{Dependence of the modulus of the multiplier associated with the localized mode with respect to the frequency for $p=0.45$.}
\label{fig:stab}
\end{figure}

{In order to explore the effect of the instability caused by the localized mode, we have simulated Eq.~(\ref{EQ:NLW}) with the perturbed stationary breather as initial condition, which has been taken in the form $u(x,0)=\tilde{u}(x,0)$, $\partial_t u(x,0)=\delta\xi(x,0)$, with $\tilde{u}(x,t)$ being the breather solution of Eq.~(\ref{EQ:NLW}), and with $\delta$ being the perturbation strength, while $\xi(x,0)$ is the corresponding component of the localized eigenmode. Interestingly, the perturbed breather starts moving with a constant speed, as can be observed in Fig.~\ref{fig:moving1}, which shows the evolution of the moving breather with $\omega=2.64$ for a perturbation $\delta=0.05$. This figure displays both the wavefunction profile $u(x,t)$ for short times and the energy density $h(x,t)$ for longer times; the latter arises from the definition of the Hamiltonian (\ref{eq:Hamiltonian}) so that $H=\int h(x)\mathrm{d}x$. For this and other simulations, a time step $\delta t=T/100\approx0.026$ was chosen, which kept the energy conserved with a relative error of $~10^{-8}$.}

\begin{figure}[!h]
\begin{center}
\begin{tabular}{cc}
\includegraphics[width=.45\textwidth]{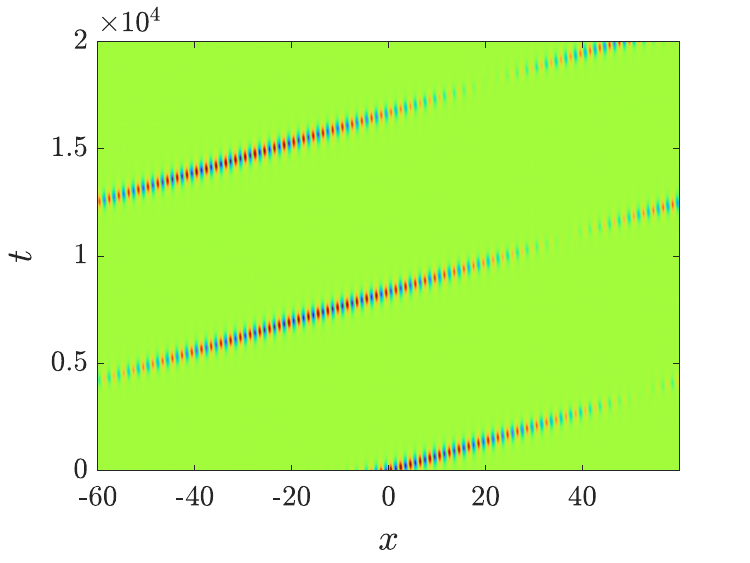} &
\includegraphics[width=.45\textwidth]{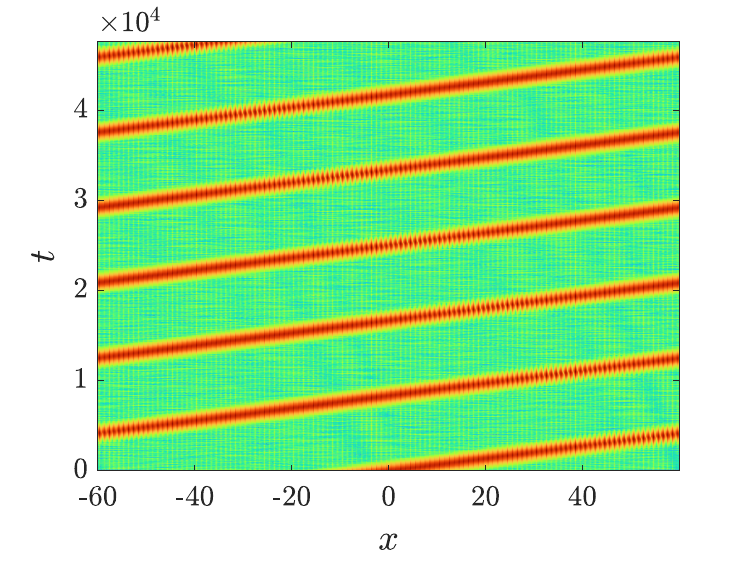} \\
\end{tabular}
\end{center}
\caption{{Evolution of the breather wavefunction $u(x,t)$ (left) and the  logarithm 
(base $10$) of the energy density $h(x,t)$ (right) with respect to time (left) for the moving breather with $\omega=2.64$ and $p=0.45$, generated by a perturbation $\delta=0.05$.}}
\label{fig:moving1}
\end{figure}

{A consequence of the smoothness of the motion can be observed in Fig.~\ref{fig:moving2} where the time-dependence of the energy density's center of mass, defined as}

\begin{equation}
    X_c=\frac{1}{H}\int xh(x)\,\mathrm{d}x, 
\end{equation}
{is plotted for breathers with $\omega=2.64$ and $\omega=2.66$ perturbed by $\delta=0.01$ and $\delta=0.05$. Remarkably, the relevant center of
mass follows a straight line, which, in turn, 
clearly indicates that the breather moves with constant velocity despite  the emitted radiation.

\begin{figure}[!h]
\begin{center}
\begin{tabular}{cc}
\includegraphics[width=.45\textwidth]{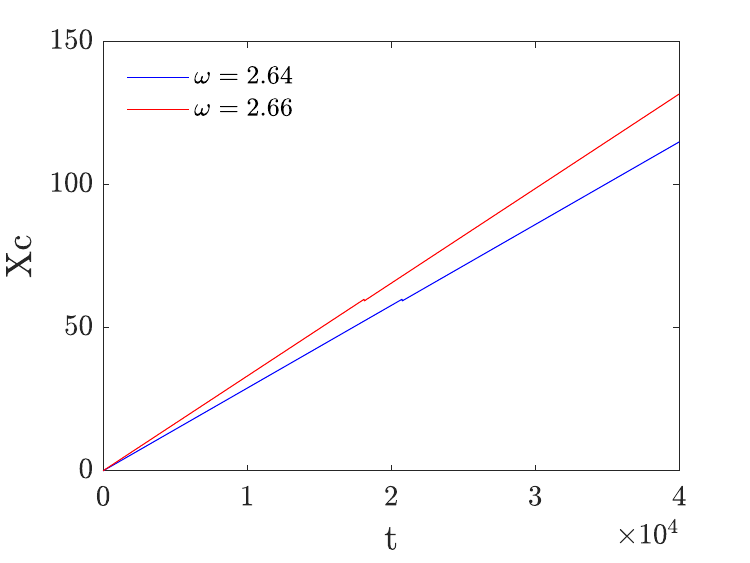} &
\includegraphics[width=.45\textwidth]{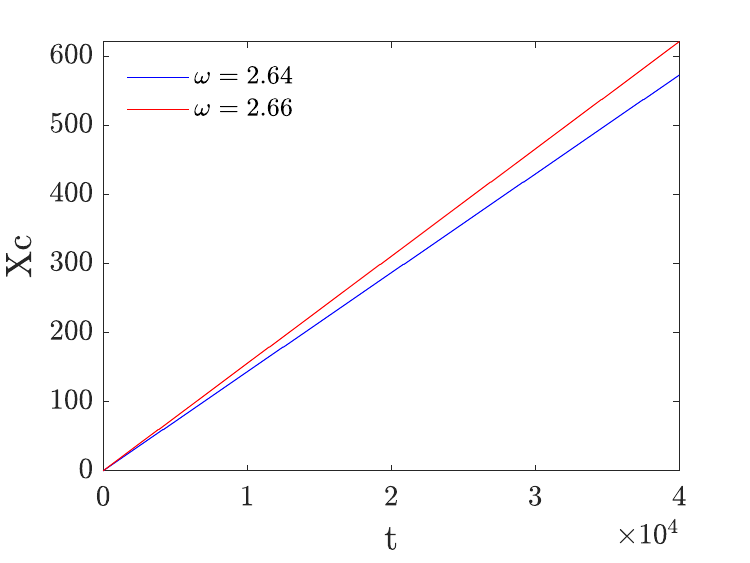} \\
\end{tabular}
\end{center}
\caption{{Evolution of the energy center of mass for the moving breathers with $\omega=2.64$ and $\omega=2.66$ obtained by adding a perturbation of amplitude $\delta=0.01$ (left) and $\delta=0.05$ (right).}}
\label{fig:moving2}
\end{figure}

{These results strongly seem to point to the instability of our stationary breather state towards moving breathers, although the localized eigenmode is different than the translational mode. One explanation of this discrepancy can rely on the fact that the projection of the localized mode onto a translational-type
  mode (i.e., one that leads to mobility) is large enough so that the perturbation is able to lead to a breather motion.} They also clearly seem to point towards the existence of exact traveling breather waveforms which would be of particular interest to identify in the so-called co-traveling frame
(traveling with the breather). Nevertheless, this is a substantial task in its own right that is deferred to future publications.

{We close this section by noting
the intriguing feature of the periodic vanishing of the field observed in Fig.~\ref{fig:moving1}. 
Probing the dynamics, we find that the vanishing frequency is $\sim8.7\times10^{-4}$. This phenomenon, which might have its origin in the fact that the unstable eigenmode is not a purely translational mode, appears over a time scale of the order of the inverse of the rate of growth of the
unstable mode. For this particular breather, the growth mode is $\sim3.7\times10^{-4}$, so its inverse is 
roughly within the same
ballpark as the observations of Fig.~\ref{fig:moving1}.
Of course, once the nonlinear dynamics of the evolution
of the instability sets in, the dynamics is less predictable,
yet here some apparent recurrence of the relevant
field vanishing seems to take place in the left panel
of the figure.
}

\section{Conclusions and Future Challenges}

In the present work we have revisited the important and
interesting topic of the potential existence of breather
type waveforms in continuous, nonlinear media. The current
pervasive impression in the nonlinear community is that
such waveforms are absent in generic nonlinear wave equations
except for the setting of integrable models such as the
sine-Gordon equation (or similarly the modified
Korteweg-de Vries equation). Nevertheless, recent significant
mathematical developments have paved the way towards
identifying such exact (up to a prescribed accuracy)
waveforms in continuum, nonlinear media, most notably
for PDEs of the Klein-Gordon type bearing (suitably
designed) spatial heterogeneity.

While these efforts have provided a theoretical backdrop
for the existence of such breather waveforms, to our knowledge,
such states have not been systematically computed previously,
nor has their stability been elucidated. This is a primary
contribution of this work, where such solutions are identified
via a Fourier space method and subsequently their Floquet
analysis has been explored, indeed as parameters are varied,
such as the parameters of the breather ($\omega$), as well
as ones of the model (such as $p$). We have generically been
able to indeed identify the theoretically established
waveforms. We have also shown that in the examples
considered, they always feature a real pair of Floquet
multipliers, associated with a spectral instability
of the relevant breathers. This instability was dynamically
explored in its own right, showcasing the potential of
these breathers for traveling when perturbed in the
pertinent unstable eigendirections of the standing breather
state.

Naturally, this work has paved the way for numerous
additional questions towards future study. One of these
certainly  concerns the direct computation (as exact solutions
in a co-traveling frame) of the traveling waveforms and the
potential examination also of their stability.
On the other hand, the instability of stationary breathers
herein also poses the question of whether there can be
other variants of this (or other similar) class of models
where heterogeneous, continuum, stationary breathers can
actually be dynamically robust and spectrally stable.
At a larger scale, given the theoretical understanding
of the requirements for establishing such structures,
it is also relevant to consider optimization problems
enabling the design of ``optimal nonlinear media'' towards,
e.g., the widest possible interval of existence of breather
waveforms. Some of these topics are presently under
consideration and relevant findings will be reported
in future publications.

\paragraph{Acknowlegdements.} PGK gratefully acknowledges discussions with Professor C.E. Wayne on the subject. MCB gratefully acknowledges prior collaboration with Dr. Christian Hamster which partly motivated this research.

\end{document}